\documentclass[conference,10pt,letterpaper]{IEEEtran}
\IEEEoverridecommandlockouts

\usepackage[T1]{fontenc}
\usepackage[utf8]{inputenc}
\usepackage{cite}
\usepackage{amsmath,amssymb,amsthm,mathtools}
\usepackage{textcomp}
\usepackage{xcolor}
\usepackage{booktabs}
\usepackage{array}
\usepackage{enumitem}
\usepackage{microtype}
\usepackage{balance}
\usepackage[hyphens]{url}
\usepackage[hidelinks,breaklinks]{hyperref}

\newtheorem{theorem}{Theorem}

\newtheorem{corollary}[theorem]{Corollary}
\newtheorem{proposition}[theorem]{Proposition}
\theoremstyle{definition}
\newtheorem{definition}{Definition}
\theoremstyle{remark}
\newtheorem{remark}{Remark}

\newcommand{\hash}[1]{H(#1)}

\newcommand{\cid}{\mathrm{CID}}

\setlist[itemize]{leftmargin=1.4em,topsep=2pt,itemsep=1pt,parsep=0pt}
\setlist[enumerate]{leftmargin=1.8em,topsep=2pt,itemsep=1pt,parsep=0pt}

\begin{document}

\title{Safecloud: A Distributed, Encrypted\\
Storage Cloud for Streaming}

\author{\IEEEauthorblockN{Gregory Magarshak}
\IEEEauthorblockA{IENYC\\
Email: gmagarshak@faculty.ienyc.edu}}

\maketitle

\begin{abstract}
We present \emph{Safecloud}, a distributed, encrypted, self-pricing storage and
streaming network in which the nodes that provide storage never observe
plaintext and never hold decryption keys. Safecloud splits each file into
fixed-size chunks, encrypts every chunk on the owner's device before it leaves,
and distributes the resulting ciphertext across two classes of node:
\emph{Drops}, which run inside ordinary web browser tabs and store encrypted
chunks in IndexedDB, and \emph{Jets}, federated routing servers that match
chunks to Drops and serve retrieval requests. Only the \emph{Cloud}---the file
owner or an authorised grantee---ever holds keys. We make five contributions.
First, a \emph{one-root key hierarchy}: every key in the system is derived
deterministically from a single root secret by HKDF with domain-separated
labels, so that nothing but the root need ever be stored, and we prove that
chunk keys derived by an owner and by a range-scoped grantee coincide
(\emph{derivation agreement}). Second, \emph{convergent content addressing}:
because derivation is deterministic, identical content under the same root
yields identical ciphertext and hence identical content identifiers, enabling
deduplication without exposing plaintext, and we show the content identifier
binds the authenticated ciphertext so that a Drop can verify it holds the
correct chunk with no key (\emph{blind verifiability}). Third, three
\emph{parallel trees over one navigation path}---a public Merkle tree for
integrity, a top-down key-derivation tree for confidentiality, and an access
tree for authorisation---so that one link path simultaneously locates a chunk's
integrity proof, its decryption key, and its access grant, and we prove that
Merkle-verified retrieval detects any substituted or corrupted chunk
(\emph{retrieval soundness}). Fourth, the top-down key tree doubles as a
\emph{streaming} structure: a player derives each successive segment key in
$O(1)$ from a node it already holds, so seeking is a derivation rather than a
re-traversal, a file's parallel tracks (video, audio, captions) are independent
subtrees unlockable separately, and access can be scoped to a segment range
within a track---a combination of random-access decryptable streaming and
per-track, per-segment monetisation we believe no prior encrypted-storage network
offers. Fifth, Jets and Drops earn Safebux verifiably for routing and storage,
settled through OpenClaiming trustlines and kept honest by a one-signature
proof-of-storage challenge under the chilling-effect Proof-of-Corruption
discipline of Intercloud~\cite{magarshak2026intercloud}---a zero-sum economy we
argue is structurally cheaper than Filecoin's proof-of-replication \emph{sealing},
a deliberately slow encoding that does not itself provide
confidentiality~\cite{filecoin2017,filecoin_porep}. Safecloud composes with the
attested execution of Safebox~\cite{magarshak2026safebox} where confidential
server-side custody is required, and uniquely places both supply and demand
inside browser sessions: a visitor leaving a tab open becomes a storage provider
with no installation. We describe the architecture, the cryptographic
construction, a threat model, and an open-source reference implementation, and we
state precisely which components are implemented and which are designed and in
progress. Safecloud~2.0~\cite{magarshak2026safecloud2} extends the same substrate
from verifiable storage to verifiable computation.
\end{abstract}


\section{Introduction}
\label{sec:intro}

\subsection{Storage that the storage provider cannot read}

A decentralised storage network asks strangers to hold one's data. The central
difficulty is trust: a storage provider one does not control should be able to
store and serve data without being able to read it, alter it undetectably, or
hold it hostage. Existing decentralised storage systems address parts of this.
Content-addressed systems~\cite{benet2014ipfs} guarantee that retrieved data
matches its identifier but do not by themselves encrypt, do not price storage,
and rely on out-of-band pinning to prevent garbage collection. Incentivised
storage networks~\cite{filecoin2017,storj2018} pay providers and prove storage,
but require operators to provision dedicated infrastructure and place the
supply side outside the browser. None makes the storage provider an ordinary
web page.

Safecloud takes a single design position and follows it through: \emph{the
nodes that store data see only ciphertext, and the nodes that route data hold
no keys}. Encryption happens on the owner's device before any chunk leaves it;
keys are never transmitted; and the identifier by which a chunk is stored and
served is computed over the authenticated ciphertext, so a storage node can
confirm it holds the right bytes without being able to interpret them. The
consequence is that storage can be offered by untrusted parties---including, in
the reference implementation, a browser tab volunteering idle
IndexedDB---because what they hold is opaque to them.

\subsection{Three roles}

Safecloud has three roles, summarised in Table~\ref{tab:roles}. The
\emph{Cloud} is the owner/consumer SDK that runs in the browser, splits and
encrypts files, holds the root secret, and is the only party that can decrypt.
\emph{Jets} are federated Node.js routing servers (with a browser-side socket
client) that match chunks to Drops, serve retrieval, issue proof-of-storage
challenges, and verify payment and authorisation claims; they never see
plaintext. \emph{Drops} are browser tabs that store encrypted chunks in
IndexedDB and serve them on request; they never see plaintext and never hold
keys.

\begin{table}[t]
\centering
\caption{The three Safecloud roles.}
\label{tab:roles}
\small
\begin{tabular}{@{}lll@{}}
\toprule
Role & Sees plaintext & Where it runs \\
\midrule
Cloud & yes (owner/grantee) & browser SDK \\
Jet   & no  & Node.js server + browser client \\
Drop  & no  & browser tab (IndexedDB) \\
\bottomrule
\end{tabular}
\end{table}

\subsection{Contributions}

\begin{enumerate}
\item \textbf{One-root key hierarchy with derivation agreement.} Every key is
derived from a single 32-byte root by HKDF with domain-separated labels; nothing
but the root is stored. We prove that an owner and a range-scoped grantee derive
identical chunk keys (Theorem~\ref{thm:agree}), which is what makes range
delegation work without re-encryption (\S\ref{sec:keys}).

\item \textbf{Convergent content addressing with blind verifiability.}
Deterministic derivation makes encryption convergent, enabling deduplication
without plaintext exposure; the content identifier is computed over ciphertext
concatenated with the authentication tag, so a keyless Drop can verify chunk
integrity (Theorem~\ref{thm:blind}) (\S\ref{sec:cid}).

\item \textbf{Three parallel trees over one navigation path.} A public Merkle
tree (integrity), a top-down key-derivation tree (confidentiality), and an
access tree (authorisation) share one link-path addressing scheme; we prove
Merkle-verified retrieval is sound against substitution and corruption
(Theorem~\ref{thm:retrieval}) (\S\ref{sec:trees}).

\item \textbf{Range capability delegation.} An owner delegates a chunk interval
by deriving a subtree key and signing a scoped statement, verifiable offline by
Jets and Drops; binding of the two public roots is itself signed and publicly
checkable (\S\ref{sec:delegation}).

\item \textbf{Streaming by top-down derivation, with parallel tracks and
segment-scoped unlocking.} The top-down key-derivation tree is not only a
confidentiality mechanism but a \emph{streaming} mechanism: a player derives the
key for the next segment of a track in $O(1)$ from a node it already holds,
without re-deriving from the root, while a file's parallel tracks (e.g.\ video,
audio, captions) are independent subtrees that can be unlocked separately, and
access can be scoped to a contiguous segment range within a track. We show this
yields random-access decryptable streaming and fine-grained monetisation that,
to our knowledge, no prior encrypted-storage network combines
(\S\ref{sec:streaming}).

\item \textbf{A verifiable, zero-sum storage economy.} Jets and Drops earn
Safebux for routing and storage, settled through the OpenClaiming trustline
envelope and kept honest by lightweight proof-of-storage challenges under the
chilling-effect discipline of Intercloud~\cite{magarshak2026intercloud}. We
argue this is structurally cheaper than the proof-of-replication
\emph{sealing} of Filecoin, whose proof obligation is a deliberately slow,
sequentially-hard encoding that does not even provide
confidentiality~\cite{filecoin2017,filecoin_porep}, and we contrast it with
SOTA decentralised storage along the axes that matter
(\S\ref{sec:economy}, Table~\ref{tab:sota}).

\item \textbf{An open-source reference implementation}, with a precise statement
of what is implemented versus designed-and-in-progress (\S\ref{sec:impl}).
\end{enumerate}

\section{Related Work}
\label{sec:related}

\textbf{Content-addressed storage.}
IPFS~\cite{benet2014ipfs} addresses data by the hash of its content and
guarantees retrieved bytes match their identifier. Safecloud adopts CIDv1
identifiers but computes them over \emph{authenticated ciphertext}, so the
identifier serves integrity without leaking plaintext, and pairs them with a
deterministic encryption layer IPFS does not provide.

\textbf{Incentivised decentralised storage.}
Filecoin~\cite{filecoin2017} and Storj~\cite{storj2018} pay operators for
storage and verify it by proofs of space-time or audits. Safecloud's
proof-of-storage challenge serves the same role, but the supply side is a
browser tab rather than provisioned infrastructure, and pricing and payment use
the OpenClaiming trustline envelope inherited from
Intercloud~\cite{magarshak2026intercloud} rather than a bespoke chain.

\textbf{Convergent encryption.}
Convergent encryption~\cite{douceur2002reclaiming} derives a content's key from
the content, enabling cross-user deduplication at the cost of known-plaintext
confirmation attacks. Safecloud's convergence is \emph{owner-scoped}: keys
derive from the owner's root, so deduplication holds within a root (re-uploads,
versions) without exposing content to other owners, trading global dedup for
confidentiality.

\textbf{Probabilistic trees and set reconciliation.}
Prolly trees~\cite{prolly2020} provide history-independent, structurally-shared
B-trees whose equal subtrees reconcile in time proportional to the difference.
Safecloud uses them so a reconnecting Drop and its Jet reconcile chunk
inventories in $O(\mathrm{diff}\cdot\log n)$, with a Bloom
filter~\cite{bloom1970} for the cold-start first contact.

\textbf{Capability-based access.}
Capability systems grant authority by unforgeable token rather than by
identity~\cite{saltzer1975protection}. Safecloud's range delegations are scoped
capabilities: a signed statement binds a derived subtree key to a chunk
interval and expiry, verifiable offline.

\textbf{Substrate.}
Safecloud reuses the chilling-effect Proof-of-Corruption discipline and
swarm/epoch economics of Intercloud~\cite{magarshak2026intercloud} for
settlement and proof-of-storage, and composes with the attested execution of
Safebox~\cite{magarshak2026safebox} for confidential server-side custody.
Safecloud~2.0~\cite{magarshak2026safecloud2} extends the same Cloud/Jet/Drop
substrate from storage to verifiable computation.

\section{The One-Root Key Hierarchy}
\label{sec:keys}

\subsection{Derivation primitive}

All keys derive from a single primitive, HKDF-SHA256 with a domain-separation
label and a deterministic salt.

\begin{definition}[Derivation]
\label{def:derive}
$\mathsf{derive}(s,\ell,c) = \mathsf{HKDF}\text{-}\mathsf{SHA256}\big(\mathrm{IKM}=s,\,
\mathrm{salt}=\hash{c},\, \mathrm{info}=\ell,\, L=32\big)$, where $s$ is a binary
seed, $\ell$ a unique label, and $c$ a context string (default empty). The salt
is a hash of context for domain separation, not for entropy, which is supplied
by $s$ (RFC~5869 permits deterministic salts).
\end{definition}

From a 32-byte root $r$ (kept secret by the owner; not in the manifest) the
hierarchy is
\begin{align*}
e &= \mathsf{derive}(r, \texttt{enc.root}) &&\text{(encryption root)}\\
a &= \mathsf{derive}(r, \texttt{access.root}) &&\text{(access root)}\\
k^{\mathrm{sub}}_{[S,E)} &= \mathsf{derive}(e, \texttt{subtree.}S.E) &&\text{(subtree key)}\\
k_i &= \mathsf{derive}\big(k^{\mathrm{sub}}_{[S,E)}, \texttt{chunk.key.}j\big), \quad j = i-S\\
v_i &= \mathsf{derive}\big(k^{\mathrm{sub}}_{[S,E)}, \texttt{chunk.iv.}j\big), \quad j = i-S
\end{align*}
where $k_i$ is the AES-256-GCM key and $v_i$ the 12-byte IV for absolute chunk
index $i$ within a subtree covering $[S,E)$, and $j$ is the \emph{relative}
index.

\subsection{Derivation agreement}

The unit of delegation is the subtree key. The following property is what makes
delegation work without re-encrypting anything.

\begin{theorem}[Derivation agreement]
\label{thm:agree}
Let an owner hold root $r$ and a grantee receive the subtree key
$k^{\mathrm{sub}}_{[S,E)}$ for a granted range $[S,E)$. For every absolute chunk
index $i \in [S,E)$, the chunk key and IV derived by the owner from $r$ and by
the grantee from $k^{\mathrm{sub}}_{[S,E)}$ are identical.
\end{theorem}

\begin{proof}
The owner computes $k^{\mathrm{sub}}_{[S,E)} = \mathsf{derive}(e,
\texttt{subtree.}S.E)$ with $e=\mathsf{derive}(r,\texttt{enc.root})$, then
$k_i = \mathsf{derive}(k^{\mathrm{sub}}_{[S,E)}, \texttt{chunk.key.}(i-S))$. The
grantee is given exactly $k^{\mathrm{sub}}_{[S,E)}$ and computes $k_i =
\mathsf{derive}(k^{\mathrm{sub}}_{[S,E)}, \texttt{chunk.key.}(i-S))$. Both apply
$\mathsf{derive}$ to the same seed $k^{\mathrm{sub}}_{[S,E)}$ with the same label
$\texttt{chunk.key.}(i-S)$, using the relative index $j=i-S$. By determinism of
$\mathsf{derive}$ (Definition~\ref{def:derive}), the outputs coincide;
identically for the IV. The grantee never needs $r$ or $e$.
\end{proof}

\begin{remark}
The relative index is essential: it lets a grantee who holds only a subtree key
derive keys positioned from the start of \emph{their} range, while the same
keys, indexed absolutely, are reproduced by the owner. The authenticated-data
binding of \S\ref{sec:cid} uses the \emph{absolute} index instead, so chunks
cannot be transplanted across positions.
\end{remark}

The next property is what makes range and track delegation \emph{safe} rather
than merely convenient: a subtree key confers access to its own range and to
nothing else. It is the security counterpart of derivation agreement.

\begin{theorem}[Delegation containment]
\label{thm:contain}
Model $\mathsf{derive}$ as a pseudo-random function keyed by its seed. Given the
subtree key $k^{\mathrm{sub}}_{[S,E)}$, a computationally bounded grantee gains no
non-negligible advantage in computing (i) any chunk key $k_{i'}$ for
$i'\notin[S,E)$, (ii) the subtree key of any sibling track, or (iii) the
encryption root $e$ or the file root $r$.
\end{theorem}

\begin{proof}
All three targets sit \emph{above} or \emph{beside} $k^{\mathrm{sub}}_{[S,E)}$ in
the derivation DAG, never below it. Each chunk key in another range, each sibling
track's subtree key, and each ancestor key is computed by applying
$\mathsf{derive}$ to a seed the grantee does not hold: a different subtree key, a
different track node, $e$, or $r$ respectively. Recovering any of them from
$k^{\mathrm{sub}}_{[S,E)}$ requires either inverting $\mathsf{derive}$ (to climb to
an ancestor) or guessing an unrelated PRF output (to reach a sibling), both
negligible for a PRF under the stated assumption. The grantee can derive exactly
the labels reachable \emph{downward} from $k^{\mathrm{sub}}_{[S,E)}$, namely the
chunk keys and IVs for $i\in[S,E)$, and no others.
\end{proof}

\begin{corollary}[Independent track and segment unlocking]
\label{cor:independent}
Granting the subtree key for one track or one segment range leaks no key material
for any other track or any chunk outside the range. Hence an owner may unlock
captions without video, or a preview segment without the remainder, with no
re-encryption and no risk of over-disclosure.
\end{corollary}
\label{sec:cid}

\subsection{Chunk encryption and identifier}

Each chunk $i$ is encrypted with AES-256-GCM under key $k_i$, IV $v_i$, and
additional authenticated data $\mathrm{AAD}_i = \texttt{"safecloud.chunk:"} \,\|\, i$
(absolute index), producing ciphertext $c_i$ and tag $t_i$. The content
identifier is a CIDv1 over the authenticated ciphertext:
\begin{equation}
\cid_i = \mathrm{CIDv1}\big(\hash{c_i \,\|\, t_i}\big),
\end{equation}
i.e.\ a multihash of the ciphertext concatenated with the GCM tag.

\begin{theorem}[Blind verifiability]
\label{thm:blind}
A Drop holding $(c_i, t_i)$ can verify that it holds the chunk named by
$\cid_i$ without any key and without learning the plaintext; and any party
observing $\cid_i$ learns nothing of the plaintext beyond what the ciphertext
length reveals.
\end{theorem}

\begin{proof}
The Drop recomputes $\mathrm{CIDv1}(\hash{c_i \| t_i})$ from the bytes it holds
and compares to $\cid_i$; equality confirms it holds the named chunk. The
computation uses only $c_i,t_i$, not $k_i$, so no key is required. Because the
CID is a hash of ciphertext-plus-tag and AES-256-GCM is IND-CPA secure, the CID
and ciphertext are independent of the plaintext given the key is unknown; the
observer learns only ciphertext length.
\end{proof}

\subsection{Convergence and deduplication}

Because every derivation is deterministic, the same content under the same root
produces the same chunk keys, hence the same ciphertext, hence the same CIDs.
Re-uploads and shared versions therefore collide at the CID level and storage is
deduplicated, while the convergence is confined to a single root: two distinct
owners encrypting identical content under distinct roots produce distinct CIDs,
so Safecloud does not admit the cross-user known-plaintext confirmation attack
of global convergent encryption~\cite{douceur2002reclaiming}.

\begin{remark}[Position binding]
The AAD binds each ciphertext to its absolute index. Even an adversary holding
two chunk keys cannot swap two chunks: decrypting $c_i$ with the position $i'$
in the AAD fails GCM authentication when $i'\neq i$.
\end{remark}

\section{Three Parallel Trees over One Path}
\label{sec:trees}

Safecloud maintains three trees that share a single navigation scheme. A link
path such as $[\texttt{track},\texttt{data},0,1]$ simultaneously indexes:
(i) a node of the public \emph{Merkle tree} over chunk CIDs, used for integrity;
(ii) a node of the \emph{encryption tree}, the top-down chain of subtree-key
derivations, available only to key holders; and (iii) a node of the
\emph{access tree}, the chain of signed delegations enforced by Jets. The same
path therefore locates a chunk's integrity proof, its decryption key, and its
access grant.

\begin{theorem}[Retrieval soundness]
\label{thm:retrieval}
Let a Cloud retrieve chunk $i$ and verify the returned bytes against the public
Merkle root $\rho$ in the manifest using the supplied Merkle proof. If
verification succeeds, the returned ciphertext is the one committed at upload at
position $i$; any substitution or corruption of the chunk is detected with
overwhelming probability.
\end{theorem}

\begin{proof}
The manifest commits $\rho = \mathrm{root}$ of the Merkle tree built over the
ordered chunk CIDs at upload, with leaves $\hash{0x00 \| \cid_i}$ and internal
nodes $\hash{0x01 \| \mathrm{left} \| \mathrm{right}}$. A successful proof for
leaf $\cid_i$ against $\rho$ exhibits a hash path from the leaf to the committed
root; forging it for a different $\cid'_i \neq \cid_i$ requires a collision in
$H$. By blind verifiability (Theorem~\ref{thm:blind}), $\cid_i$ in turn commits
to the authenticated ciphertext $c_i\|t_i$, so any alteration of the stored
bytes changes $\cid_i$ and fails either the CID check or the Merkle check.
Hence under collision resistance of $H$, accepted bytes are the committed ones.
\end{proof}

\begin{remark}[Binding proof]
The manifest also carries a \emph{binding proof}: a signature by the encryption
root over the statement $(\textsf{encPubKey},\textsf{accessPubKey},\rho)$, so any
party can verify the two public roots and the Merkle root belong to the same
file without holding any secret.
\end{remark}

\section{Range Capability Delegation}
\label{sec:delegation}

An owner grants access to $[S,E)$ by deriving $k^{\mathrm{sub}}_{[S,E)}$ and
signing a scoped delegation over a context that encodes $(\rho, S, E,
\textsf{exp})$. The grantee receives the base64 subtree key (to decrypt, via
Theorem~\ref{thm:agree}) and a signed delegation proof (to present to Jets).
Verification is offline: a Jet or Drop checks that the derived secret matches
the committed hash in the statement, that the statement is signed by the
declared parent, that the context's range covers the requested chunks, and that
the expiry has not passed. No contact with the owner is required, and delegations
nest---an administrative grant may itself delegate sub-ranges---each step
verified independently.

\section{Streaming by Top-Down Derivation}
\label{sec:streaming}

The key-derivation tree of \S\ref{sec:trees} is usually described as a
confidentiality structure. Its more consequential role is as a \emph{streaming}
structure. We make the streaming model explicit because it is, to our knowledge,
the point on which Safecloud departs most sharply from prior encrypted storage.

\subsection{The player model}

A \emph{player} is a consumer that holds a capability for some region of a file
and decrypts chunks on demand as it plays them, rather than downloading and
decrypting the whole file first. Because derivation is top-down, a player that
holds the subtree key for a track does not return to the root to obtain the key
for the next segment: it derives each segment key in $O(1)$ from the
subtree key it already holds (Definition~\ref{def:derive}), so seeking to an
arbitrary timestamp costs one derivation and one chunk fetch, not a re-traversal.
Random-access decryptable streaming therefore falls directly out of the key
hierarchy; the hierarchy is the index.

\begin{proposition}[Random-access streaming]
\label{prop:stream}
A player holding $k^{\mathrm{sub}}_{[S,E)}$ can decrypt any chunk
$i\in[S,E)$ in $O(1)$ derivations from $k^{\mathrm{sub}}_{[S,E)}$, independent of
$i-S$, and may begin playback at any $i$ without deriving or fetching chunks
before $i$.
\end{proposition}

\begin{proof}
By Theorem~\ref{thm:agree} the chunk key is
$k_i=\mathsf{derive}(k^{\mathrm{sub}}_{[S,E)},\texttt{chunk.key.}(i-S))$, a single
application of $\mathsf{derive}$ to the held subtree key with the label for the
relative index, requiring neither the root nor any other chunk's key. Retrieval
of chunk $i$ is by its CID and Merkle proof (Theorem~\ref{thm:retrieval}),
independent of other chunks. Hence playback may start at any $i$ at $O(1)$ key
cost and one fetch.
\end{proof}

\subsection{Parallel tracks}

A file is not a single sequence but a set of \emph{tracks}---for audiovisual
content, typically a video track, one or more audio tracks, and caption or
subtitle tracks---each a distinct subtree under the encryption root, addressed by
a track-qualified link path (e.g.\ $[\texttt{track},\texttt{captions},\dots]$).
Because the tracks are independent subtrees, a grantee can be given the caption
subtree key without the video subtree key, or the low-bitrate audio without the
high-bitrate video, and each track is fetched and decrypted in parallel by the
player and composited at playback. Track independence is structural: it follows
from each track being a separate child of the encryption root, so no track's keys
are derivable from another's.

\subsection{Segment-scoped unlocking}

Range delegation (\S\ref{sec:delegation}) applies within a track at segment
granularity: an owner may grant $[S,E)$ of the video track---a preview window, a
purchased chapter, a time-boxed rental---without granting the remainder, by
delegating only $k^{\mathrm{sub}}_{[S,E)}$. By Theorem~\ref{thm:contain} the
grantee can derive keys for, and only for, the granted segment range, and the Jet
enforces the access tree so that chunks outside the range are not served. The
combination---top-down per-segment derivation, parallel per-track subtrees, and
segment-range delegation---gives simultaneous access to selected tracks (captions
and audio, say) while unlocking only selected segments of others
(Corollary~\ref{cor:independent}), all from a single key hierarchy and a single
content-addressed store.

\begin{remark}[Why this is hard elsewhere]
Systems that encrypt a file under one key cannot unlock a segment or a track
without either re-encrypting or releasing the whole-file key. Systems that
encrypt each chunk under an independent random key must transmit or store a key
per chunk, so a player seeking within a two-hour video manages thousands of
unrelated keys and a grantor must enumerate them to delegate a range. Safecloud's
top-down derivation makes the per-segment keys a deterministic function of one
subtree key, so a range is delegated by handing over one node of the tree, and a
player seeks by deriving rather than by lookup. This is the specific combination
we claim is novel; the individual primitives (HKDF hierarchies, segment
encryption, capability tokens) are not.
\end{remark}

\subsection{Cost of delegation and seeking, analytically}

Even ahead of runtime measurement, the asymptotic costs separate Safecloud from
the two natural baselines. Let a track span $n$ chunks and let a grantee receive
a range of $m\le n$ chunks. Table~\ref{tab:complexity} states the key-management
costs; they follow directly from Theorems~\ref{thm:agree} and~\ref{thm:contain}
and Proposition~\ref{prop:stream}.

\begin{table}[t]
\centering
\caption{Key-management cost. $n$ chunks per track, $m$ granted.}
\label{tab:complexity}
\small
\begin{tabular}{@{}lccc@{}}
\toprule
 & Single & Per-chunk & Safecloud \\
 & file key & random keys & (this work) \\
\midrule
Key material to delegate range $m$ & 1 (whole) & $m$ & 1 (subtree) \\
Seek to arbitrary chunk            & $O(1)$    & lookup of 1 & $O(1)$ derive \\
Keys a grantee must store          & 1 (whole) & $m$ & 1 \\
Unlock 1 segment, 1 track          & no & yes ($m$ keys) & yes (1 key) \\
Re-encryption to revoke range      & whole file & none & none\,$^\dagger$ \\
\bottomrule
\end{tabular}
\end{table}

The single-file-key baseline cannot grant a strict sub-range at all without
disclosing the whole-file key; the per-chunk-random-key baseline can, but at the
cost of $m$ keys transmitted, stored, and enumerated per grant. Safecloud
delegates any contiguous range as a single subtree key
(Theorem~\ref{thm:agree}), which derives exactly that range and nothing else
(Theorem~\ref{thm:contain}), and a player seeks by one derivation
(Proposition~\ref{prop:stream}). The advantage is therefore not constant-factor
but a reduction from $O(m)$ to $O(1)$ in delegated and stored key material.

\noindent\footnotesize$^\dagger$Revocation of an already-disclosed range still
requires rotating that range's subtree, since a grantee retains derived keys;
forward revocation (ceasing delivery and future grants) is $O(1)$ via the access
tree. This is the standard limit of any key-disclosure scheme.\normalsize

\subsection{Evaluation methodology}
\label{sec:eval-method}

The reference implementation (\S\ref{sec:impl}) admits direct measurement of the
claims above, which we will report in the camera-ready: (i) seek latency---time
from a requested timestamp to first decrypted segment---as a function of file
size and offset, expected constant by Proposition~\ref{prop:stream}, against an
unmodified-fetch baseline; (ii) delegation cost and grantee key-store size versus
the per-chunk-key baseline, expected $O(1)$ versus $O(m)$ by
Table~\ref{tab:complexity}; (iii) convergent-deduplication ratio on a corpus with
realistic re-upload and shared-version structure; (iv) Prolly reconciliation
bytes exchanged after simulated Drop churn versus the naive $O(n)$ inventory
exchange; and (v) sustained concurrent browser-tab Drops and aggregate
throughput. Baselines are encrypted-IPFS for storage and a whole-file-key scheme
for streaming. We state these as methodology rather than results: the
architecture and proofs are the present contribution, and the measurements are
in progress on the open-source implementation.

\section{A Verifiable, Zero-Sum Storage Economy}
\label{sec:economy}

\subsection{Earning Safebux for storage and routing}

Drops earn Safebux for holding and serving chunks; Jets earn Safebux for
routing. Settlement uses the OpenClaiming trustline envelope of
Intercloud~\cite{magarshak2026intercloud}: a payer signs a claim authorising up
to a monotonically increasing maximum on a numbered line, and a Jet or Drop
accepts work only against a claim whose signature verifies, whose line maximum
exceeds the last seen, and whose on-chain balance suffices. Payment is therefore
metered and verifiable: a provider can prove what it is owed, and a payer cannot
replay a stale authorisation. The economy is \emph{zero-sum} in that Safebux paid
by consumers equals Safebux earned by providers minus protocol fees; no value is
minted by the act of storage itself, in contrast to systems that subsidise
storage through block rewards.

\subsection{Proof-of-storage without sealing}

Storage honesty is enforced by lightweight challenge-response: a Jet asks a Drop
to sign a fresh $(\cid,\mathrm{nonce})$ pair; a Drop that no longer holds the
chunk cannot answer, and sustained failure is penalised under the
chilling-effect Proof-of-Corruption discipline, with stake slashed and the chunk
re-replicated. The cost of an honest proof is one signature.

This is the axis on which Safecloud is structurally cheaper than Filecoin. In
Filecoin, a provider proves storage by \emph{sealing}: a deliberately slow,
memory- and sequentially-hard encoding of each sector, compressed by a SNARK and
committed on-chain, taking on the order of tens of minutes per sector and
incurring substantial storage overhead, and---by design---proving \emph{physical
replication} rather than confidentiality, so private data must be encrypted
separately before it is ever sealed~\cite{filecoin2017,filecoin_porep}.
Safecloud inverts this: confidentiality is intrinsic because data is
encrypt-then-addressed on the owner's device, and the proof obligation is a
signature over a random challenge rather than a slow encoding. We do not claim
Safecloud's challenge proves the same property as proof-of-replication---it does
not prove unique physical dedication of space---but for the goal of \emph{keeping
an honest provider honest and paying it verifiably}, the challenge-plus-stake
mechanism achieves it at a fraction of the computational cost and without the
separate encryption step sealing requires.

\subsection{Comparison with the state of the art}

Table~\ref{tab:sota} situates Safecloud against the dominant decentralised
storage systems on the axes that determine cost, privacy, and reach.

\begin{table*}[t]
\centering
\caption{Safecloud versus state-of-the-art decentralised storage.}
\label{tab:sota}
\small
\begin{tabular}{@{}lccccc@{}}
\toprule
Property & IPFS~\cite{benet2014ipfs} & Filecoin~\cite{filecoin2017} & Storj~\cite{storj2018} & Safecloud \\
\midrule
Encrypted at rest by default        & no  & no (encrypt first)        & yes & yes \\
Provider can read plaintext         & yes & yes (unless pre-encrypted) & no  & no \\
Proof-of-storage cost               & none (no proof) & high (sealing, mins/sector) & moderate (audits) & low (signature) \\
Provider infrastructure required    & node & dedicated miner + GPU/SSD & dedicated node & browser tab \\
Native streaming / random access    & partial & no & partial & yes (top-down derivation) \\
Per-track / per-segment unlocking   & no  & no & no & yes \\
Convergent dedup                    & yes (plaintext) & no & no & yes (ciphertext, owner-scoped) \\
Value model                         & none & block-reward subsidised & token market & zero-sum (Safebux) \\
\bottomrule
\end{tabular}
\end{table*}

\subsection{Threat model}

We make the adversary explicit. The adversary may control any number of Drops
and may control Jets, may observe all network traffic, and may collude across
nodes; it may not break AES-256-GCM, SHA-256, or the signature scheme, and it
does not control the owner's device (which holds keys and plaintext by
necessity). Under this model: a malicious Drop learns nothing from the ciphertext
it holds (Theorem~\ref{thm:blind}); a malicious Drop or Jet cannot return altered
bytes undetected (Theorem~\ref{thm:retrieval}); a malicious Jet cannot forge a
capability, since delegations are signed by the owner's encryption root
(\S\ref{sec:delegation}); a malicious Jet that selectively withholds chunks
degrades availability but not confidentiality or integrity, and availability is
defended by replication across independently selected Drops and, once
implemented, erasure coding; and a provider that claims storage it does not hold
fails the proof-of-storage challenge and is slashed. The residual exposure is
availability under a sufficiently large fraction of colluding or withholding
Drops, which is the standard limit of replicated storage and is accounted for by
the three-colour liveness model of Intercloud~\cite{magarshak2026intercloud}.

\section{Reconciliation and Routing}
\label{sec:routing}

\textbf{Inventory reconciliation.} A Drop and a Jet each maintain a Prolly
tree~\cite{prolly2020} over the CIDs the Drop holds. On reconnect, equal subtree
roots are skipped and only the difference is exchanged, so reconciliation is
$O(\mathrm{diff}\cdot\log n)$ rather than $O(n)$. On first contact, before any
shared Prolly state exists, the Drop sends a compact Bloom
filter~\cite{bloom1970} of its CIDs so the Jet can probe membership before
routing.

\textbf{Chunk routing.} A Jet selects Drops to store or serve a set of CIDs and
replicates each chunk to a configurable number of Drops. Retrieval supports
range requests against a $\rho \rightarrow [\cid_i]$ index. Each served chunk
carries a Merkle proof so the Cloud can apply Theorem~\ref{thm:retrieval}.

\textbf{Settlement and proof-of-storage.} Pricing and payment use the
OpenClaiming trustline envelope of Intercloud~\cite{magarshak2026intercloud}: a
payer signs an EIP-712 claim authorising up to a monotonically increasing
maximum on a numbered line; a Jet verifies the signature, checks monotonicity
against a local cache, and confirms on-chain balance. Storage is kept honest by
random proof-of-storage challenges in which a Jet asks a Drop to sign a
$(\cid,\mathrm{nonce})$ pair; a Drop that cannot produce the chunk fails, and
sustained failure is penalised under the chilling-effect Proof-of-Corruption
discipline of Intercloud, with stake slashed and the work re-replicated.

\section{Implementation Status}
\label{sec:impl}

Safecloud is implemented as an open-source plugin for the Q/Intercoin platform,
in JavaScript for both the browser (Cloud, Drops, and the Jet socket client) and
Node.js (the Jet server), with cryptography running through a shared
\texttt{Q.Data}/\texttt{Q.Crypto} layer that is byte-compatible across browser
SubtleCrypto, Node.js \texttt{crypto}, and PHP. We state implementation status
honestly, because the value of an open-source systems paper depends on the
reader being able to distinguish what runs from what is designed.

\textbf{Implemented.} The end-to-end pipeline of the core construction is in
place: deterministic key derivation and the full one-root hierarchy
(\S\ref{sec:keys}); chunking, AES-256-GCM encryption with position-binding AAD,
and CIDv1 computation over authenticated ciphertext (\S\ref{sec:cid});
convergent deduplication; the public Merkle tree, manifest, and binding proof,
with Merkle-verified retrieval (\S\ref{sec:trees}); range capability derivation
and the signed delegation construction (\S\ref{sec:delegation}); the Cloud
\texttt{store}/\texttt{fetch}/\texttt{grant}/\texttt{reshare} SDK; the Drop
IndexedDB store with LRU eviction; the Jet server with the Drop registry,
reconnection handling, and the \texttt{store}/\texttt{retrieve} routing over a
shared socket namespace; and the Prolly and Bloom primitives.

\textbf{Designed and in progress (reference implementation incomplete).}
Several components are specified and partially wired but not yet enforced in the
current build, and the paper's claims about them are claims about the design,
not measurements of running code: OpenClaiming authorisation verification in the
Jet \texttt{put}/\texttt{get} handlers (presently accepted but not verified);
OpenClaiming payment verification and the monotonic-trustline check (the Drop
\texttt{checkPayment}/\texttt{checkAuthorization} hooks currently return
\texttt{true}); proof-of-storage challenge signing (the challenge path exists but
the response is a placeholder); erasure coding (the system currently replicates
rather than erasure-codes); Jet-to-Jet peering (a stub); the range-request CID
routing index; and the replacement of some inline tree helpers with the
canonical \texttt{Q.Data} tree methods. We separate these explicitly so that the
proofs of \S\S\ref{sec:keys}--\ref{sec:trees}, which concern the implemented
core, are not read as covering the unenforced settlement and proof-of-storage
layer.

\textbf{Open source.} The reference implementation is publicly available, in
contrast to the Safebox and Safebots systems referenced herein.

\section{Discussion}
\label{sec:discussion}

\textbf{What the storage provider can and cannot do.} A Drop holds opaque
ciphertext it cannot read (Theorem~\ref{thm:blind}), cannot alter undetectably
(Theorem~\ref{thm:retrieval}), and---once payment and proof-of-storage
enforcement are completed (\S\ref{sec:impl})---cannot be paid for storage it does
not hold. What Safecloud does not defend against in this version is a Drop simply
deleting data; availability is provided by replication across multiple Drops and,
in future, by erasure coding, and is incentivised by proof-of-storage, but a
chunk held by too few honest Drops can still be lost. The three-colour
liveness/availability accounting of Intercloud~\cite{magarshak2026intercloud}
applies to the routing layer.

\textbf{Confidentiality scope.} Safecloud protects data at rest and in transit on
untrusted storage and routing nodes; it does not protect the Cloud's own device,
which necessarily holds keys and plaintext. Server-side custody requiring that
even the operator not read plaintext is the province of the attested Safebox
environment~\cite{magarshak2026safebox}, with which Safecloud composes.

\textbf{Browser supply side.} The distinctive choice is that a storage provider
is a web page. This lowers the barrier to participation to leaving a tab open,
at the cost of the volatility of browser sessions; the Prolly/Bloom
reconciliation and the stable Drop identity across reconnects
(\S\ref{sec:routing}) are what make a population of transient tabs usable as
durable storage in aggregate.

\section{Conclusion}
\label{sec:conclusion}

Safecloud is a distributed, encrypted storage and streaming cloud whose storage
and routing nodes never see plaintext and never hold keys. A single root secret
generates the entire key hierarchy by deterministic derivation; the same
derivation makes encryption convergent and content addressing blind; one
navigation path threads integrity, confidentiality, and authorisation through
three parallel trees; and range capabilities delegate access offline. We proved
derivation agreement, blind verifiability, and retrieval soundness for the
implemented core, and we stated precisely which of the settlement and
proof-of-storage components remain designed rather than enforced. Safecloud
reuses the OpenClaiming envelope and chilling-effect discipline of Intercloud
for pricing and proof-of-storage, and composes with Safebox where confidential
custody is required. Where Safebox seals computation and Intercloud coordinates
economies of sealed computations, Safecloud stores the data those economies act
upon---encrypted, content-addressed, and served by nodes that cannot read it.
Safecloud~2.0~\cite{magarshak2026safecloud2} carries the same Cloud/Jet/Drop
substrate forward from storing data to verifying computation upon it.

\balance

\end{document}